\documentclass[twocolumn,notitlepage,nofootinbib,superscriptaddress,10pt]{revtex4-2}
\pdfoutput=1

\usepackage{amsmath,amssymb,bm}
\usepackage{amsthm}
\usepackage{graphicx}
\usepackage{hyperref}
\hypersetup{hidelinks}
\usepackage{microtype}
\usepackage{array}

\newcolumntype{L}[1]{>{\raggedright\arraybackslash}p{#1}}
\newcolumntype{C}[1]{>{\centering\arraybackslash}p{#1}}

\newcommand{\dd}{\mathrm{d}}
\newcommand{\Tr}{\operatorname{tr}}

\newcommand{\Sp}{S_p}
\newcommand{\Pc}{P_c}

\newcommand{\Ucoh}{\mathcal{U}}
\newcommand{\Ctwo}{\mathcal{C}_2}
\newcommand{\Aarrow}{\mathcal{A}_{\lambda}}

\newcommand{\Pcperp}{P_{c,\perp}}
\newcommand{\AlgIRB}{\mathfrak{A}_{\rm IRB}}

\newtheorem{proposition}{Proposition}

\begin{document}

\title{Intrinsic Pointer Basis and Irreversible Classicality from Coherence Contraction}

\author{Jos\'e J.~Gil}
\email{ppgil@unizar.es}
\affiliation{Department of Applied Physics, University of Cantabria, Av. Los Castros 48, 39005 Santander, Spain}

\date{\today}

\begin{abstract}
This work analyzes an operational route to classical behavior for reduced quantum states using the intrinsic reference basis (IRB).  Relative to a fixed physical conjugation, the IRB separates intrinsic populations from a real antisymmetric cohesion sector.  A globally bounded cohesion index is defined and its exponential contraction is proved for phase-free dephasing dynamics aligned with the IRB; for general aligned dephasing, the corresponding modulus-based coherence functional contracts at the same computable rates.  The results provide distance bounds to the IRB-diagonal description and a logarithmic upper bound on the time required to reach a prescribed experimental tolerance.  The IRB projectors constitute state-derived candidate pointer sectors, and they become dynamically stable pointer sectors when the effective dephasing generator is aligned with them and damps the relevant inter-sector coherences.  Degenerate population sectors lead naturally to block-classicality and protected intra-block coherence.  In a two-level active sector, the cohesion index equals fringe visibility, giving a direct interferometric test of the contraction law.  The construction is independent of any spacetime- or unification-emergence hypothesis and is intended as a channel-level complement to environment-induced einselection.
\end{abstract}

\maketitle

\section{Introduction}\label{sec:intro}

The suppression of phase-sensitive observables in a reduced quantum state is a central
ingredient of the quantum-to-classical transition.  In the standard decoherence account,
the robust alternatives are determined by the system--environment interaction, and their
identification generally requires either a microscopic model or process-level information
about the reduced evolution~\cite{Joos2003,Schlosshauer2007,Zurek2003}.
The present work asks a narrower and operationally precise question.  Once a reduced state
and a physical real structure are fixed, can one obtain an intrinsic decomposition that
separates population content from a cohesion sector and can one determine when an aligned
open-system dynamics contracts that cohesion sector?

Let $K$ be a fixed physical conjugation on the active Hilbert space.  In the present paper
this conjugation is part of the state-and-channel specification, in the same way that a
measurement convention, a real symmetry representation, or a secular basis is part of an
operational description.  If $K$ is selected by an additional calibration or variational
principle, the contraction statements below apply after that structure has been fixed.
Diagonalizing the $K$-real symmetric part of a density operator by an orthogonal
transformation defines the intrinsic reference basis (IRB), in which
\begin{equation}
\begin{gathered}
\rho_O=Q^{\mathsf T}\rho Q=A+\mathrm{i}N,\\
A=\operatorname{diag}(a_1,a_2,\ldots),\quad
N^{\mathsf T}=-N .
\end{gathered}
\label{eq:IRBdecomp}
\end{equation}
Here $a_1\geq a_2\geq\cdots\geq0$ and $\sum_i a_i=1$ are intrinsic populations,
whereas $n_{ij}:=N_{ij}$ ($i<j$) are intrinsic cohesion components.  This construction is
canonical relative to $K$ and to the ordering of the populations, with the residual sign
and degeneracy freedoms stated explicitly in Appendix~\ref{app:IRB}; related intrinsic
density-matrix descriptors were developed in Refs.~\cite{Gil2020Asymmetry,Gil2023Dimensionality}.

The quantity used throughout this paper is the cohesion index
\begin{equation}
\Pc^2=2\|N\|_F^2=4\sum_{i<j}n_{ij}^2,
\qquad 0\leq\Pc\leq1 .
\label{eq:Pc_intro}
\end{equation}
The upper bound is a global positivity result for density operators, proved below; it does
not follow from simultaneous saturation of pairwise $2\times2$ positivity inequalities.
This point is important in dimensions larger than two.

We study reduced Markovian dynamics whose dephasing generator is aligned with the IRB
projectors $\Pi_i=\lvert i\rangle\langle i\rvert$.  More precisely, the exact contraction
results assume a fixed family of IRB projectors and a GKSL representation in which the
Hamiltonian and dephasing operators commute with these projectors.  The off-diagonal
moduli then decay as
\begin{equation}
|\rho_{ij}(t)|=|\rho_{ij}(0)|e^{-\Gamma_{ij}t},
\qquad \Gamma_{ij}\geq0 .
\label{eq:nij_decay_intro}
\end{equation}
For Hermitian diagonal dephasing operators after removal of the diagonal Hamiltonian
phase evolution, and more generally whenever the remaining pair phases vanish, the
intrinsic components obey
$N_{ij}(t)=N_{ij}(0)e^{-\Gamma_{ij}t}$ and $\Pc$ becomes a Lyapunov diagnostic.
With $\Gamma_{\min}$ the smallest positive rate acting on an initially occupied cohesion
component, the state reaches the tolerance $\Pc\leq\varepsilon$ no later than
\begin{equation}
 t_{\rm cl}^{\rm ub}(\varepsilon)=
 \max\!\left\{0,\frac{1}{\Gamma_{\min}}
 \ln\!\left(\frac{\Pc(0)}{\varepsilon}\right)\right\} .
\label{eq:tcl_intro}
\end{equation}
Equality holds for a single active decay rate, in particular for the two-level example.

The IRB is state-derived, but pointer stability is a dynamical statement.  Accordingly,
we do not claim that a reduced state alone identifies the environmental pointer structure.
Rather, the IRB supplies candidate pointer projectors, and these projectors become stable
under the explicitly testable alignment condition on the effective generator; effective
selection additionally requires positive damping rates between the sectors being resolved.  This
formulation is complementary to einselection and remains applicable when the generator
is reconstructed experimentally rather than derived from a microscopic bath model.

The physical content of the criterion is directly testable.  In a normalized two-level
active sector the fringe visibility equals $\Pc$ exactly, independently of population
balance.  In higher dimension, pairwise fringe visibilities determine the contributing
$|n_{ij}|$, while tomography determines $\Pc$ and distance bounds to the IRB-diagonal
state.

The article is self-contained and does not rely on any claim concerning emergent
spacetime, Lorentz signature, gravity, or unification.  Section~\ref{sec:toolkit}
defines the IRB decomposition and its operational diagnostics.  Section~\ref{sec:gksl}
proves contraction under aligned dephasing.  Section~\ref{sec:threshold} gives the
classicality threshold, the classicalization-time bound, and the block-classical extension.
Sections~\ref{sec:comparison} and~\ref{sec:signatures} relate the construction to
einselection and measurable tests, respectively.

\section{Intrinsic decomposition of the reduced state}\label{sec:toolkit}

\subsection{The intrinsic reference basis}

The IRB is defined directly from the reduced density operator $\rho$ on its active support,
without reference to any external coordinate choice or system-environment model.
Starting from a fixed intrinsic conjugation $K$ (an antiunitary involution on the active
subspace specifying a ``real form'' of the quantum-state space; Appendix~\ref{app:IRB}),
any density operator decomposes as
\begin{align}
\rho &= S + \mathrm{i}T,\nonumber\\
S&:=\frac{\rho+K\rho K}{2},\quad S^{\mathsf T}=S,\nonumber\\
T&:=\frac{\rho-K\rho K}{2\mathrm{i}},\quad T^{\mathsf T}=-T,
\label{eq:SAdecomp}
\end{align}
The residual gauge freedom preserving $K$ is an orthogonal group $SO(d_{\rm act})$ acting by
similarity transformations.
The IRB fixes this gauge by diagonalizing $S$: an orthogonal matrix $Q$ is chosen so that
$Q^{\mathsf T}SQ=A=\mathrm{diag}(a_1,\ldots,a_{d_{\rm act}})$ with $a_1\ge a_2\ge\cdots$,
and the intrinsic coherences are defined by
$N:=Q^{\mathsf T}TQ$, $n_{ij}:=(N)_{ij}$ for $i<j$.
The resulting state is $\rho_O = A+\mathrm{i}N$, as in Eq.~\eqref{eq:IRBdecomp}.

The IRB is unique when the spectrum of $S$ on the active support is nondegenerate;
degenerate eigenvalues lead to pointer subspaces rather than a unique pointer basis
(Sec.~\ref{sec:degeneracies}).
Gauge-invariant quantities (the ordered set $\{a_i\}$ and any functional of $N$ that is
$SO(d_{\rm act})$-invariant, such as $\|N\|_F$) are independent of the representative $Q$
chosen within each gauge orbit.
The construction algorithm and uniqueness analysis are detailed in Appendix~\ref{app:IRB}.

Physically, the role of the IRB is that of a canonical frame: it does not represent an
external choice but is extracted from the state in a way that unambiguously separates the
diagonal sector, which carries emergent classical weights, from the antisymmetric sector,
which carries the phase-sensitive content responsible for interference.

\subsection{Populations, cohesion, and its global bound}

The population entropy, introduced under the name ``dimensional entropy'' in
Ref.~\cite{Gil2023Dimensionality}, is
\begin{equation}
\Sp:=-\sum_i a_i\ln a_i .
\label{eq:Sp}
\end{equation}
The quadratic intrinsic cohesion is
\begin{equation}
\Ucoh:=\sum_{i<j}n_{ij}^2=\frac{\|N\|_F^2}{2},
\label{eq:Ucoh}
\end{equation}
and the cohesion index is defined by
\begin{equation}
\Pc^2:=2\|N\|_F^2=4\Ucoh .
\label{eq:Pc}
\end{equation}
This definition is independent of the instantaneous distribution of populations and is
therefore appropriate for tracking loss of intrinsic cohesion when populations also change.

Positive semidefiniteness implies the useful pairwise check
\begin{equation}
n_{ij}^2\leq a_i a_j \qquad(i\ne j),
\label{eq:positivity_nij}
\end{equation}
but in dimensions greater than two these inequalities do not characterize positivity of
the full density operator.  The normalization of Eq.~\eqref{eq:Pc} is instead supported
by the following global result.

\begin{proposition}[Global bound on intrinsic cohesion]\label{prop:PcBound}
For every density operator $\rho$ and every fixed conjugation $K$, its IRB cohesion index
satisfies $0\leq\Pc\leq1$.
\end{proposition}
\begin{proof}
The Frobenius norm of the antisymmetric part is invariant under the real orthogonal
transformation that defines the IRB, so it suffices to work before diagonalizing the
symmetric part.  For a pure state $\rho=\lvert\psi\rangle\langle\psi\rvert$, write
$\lvert\psi\rangle=\bm{x}+\mathrm{i}\bm{y}$ in a $K$-real basis and use the irrelevant
global phase to impose $\bm{x}^{\mathsf T}\bm{y}=0$.  Then
$T=\bm{y}\bm{x}^{\mathsf T}-\bm{x}\bm{y}^{\mathsf T}$ and
$\|T\|_F^2=2\|\bm{x}\|^2\|\bm{y}\|^2\leq1/2$ because
$\|\bm{x}\|^2+\|\bm{y}\|^2=1$.  Hence $2\|T\|_F^2\leq1$ for pure states.
For a mixed state $\rho=\sum_k p_k\rho_k$ with pure $\rho_k$, the triangle inequality
gives $\|T\|_F\leq\sum_k p_k\|T_k\|_F\leq1/\sqrt{2}$.  Therefore
$\Pc^2=2\|N\|_F^2=2\|T\|_F^2\leq1$.
\end{proof}

We denote by $\mathcal I_{\rm act}$ the active support, with cardinality
$d_{\rm act}$; exact statements use $a_i>0$, while experimental implementations may use
a declared population threshold.

\subsection{Distance to the IRB-diagonal description and visibility}

Let $\Delta_{\rm IRB}$ be the complete dephasing channel in the fixed IRB, so that
$\Delta_{\rm IRB}[\rho_O]=A$.  Equations~\eqref{eq:Ucoh}--\eqref{eq:Pc} imply the exact
Hilbert--Schmidt identity
\begin{equation}
\big\|\rho_O-\Delta_{\rm IRB}[\rho_O]\big\|_2^2
=\|N\|_F^2=\frac{\Pc^2}{2} .
\label{eq:hsdist}
\end{equation}
For trace distance, using $\|X\|_1\leq\sqrt{\operatorname{rank}X}\|X\|_F$ yields
\begin{align}
D_{\rm tr}\!\left(\rho_O,\Delta_{\rm IRB}[\rho_O]\right)
&\leq \sqrt{\frac{r_N}{8}}\,\Pc\nonumber\\
&\leq \sqrt{\frac{d_{\rm act}}{8}}\,\Pc,
\label{eq:trdist_bound}
\end{align}
where $r_N:=\operatorname{rank}N$.
Thus contraction of $\Pc$ gives an operational bound on proximity to the IRB-diagonal
description.

In a two-path readout restricted to IRB sectors $i$ and $j$, the fringe visibility is~\cite{Englert1996}
\begin{equation}
V_{ij}=\frac{2|n_{ij}|}{a_i+a_j} .
\label{eq:visibility_general}
\end{equation}
For a normalized two-level active state ($a_1+a_2=1$), Eq.~\eqref{eq:Pc} gives the exact
identity
\begin{equation}
\Pc=2|n_{12}|=V_{12},
\label{eq:pc_visibility_balanced}
\end{equation}
without requiring balanced populations.  In a higher-dimensional state, each $V_{ij}$
is a pairwise witness; together with the populations, these visibilities reconstruct the
corresponding contributions to $\Pc$.

For comparison with standard resource-theory diagnostics, once the IRB projectors are
fixed for a dynamical test, the relative entropy of coherence
\begin{equation}
\begin{gathered}
C_{\rm rel}(\rho):=S\!\left(\Delta_{\rm IRB}[\rho]\right)-S(\rho),\\
S(\rho):=-\Tr(\rho\ln\rho).
\end{gathered}
\label{eq:Creldef}
\end{equation}
is a standard monotone under incoherent operations relative to that fixed projective
structure~\cite{Baumgratz2014,Streltsov2017}.  No claim is made that $\Pc$ is a monotone
under every incoherent CPTP map.  The result proved here is its Lyapunov contraction under
the aligned phase-free dephasing semigroups defined in Sec.~\ref{sec:gksl}.

\subsection{Conditional arrow functional}
\label{subsec:arrow}

When an additional population dynamics increases $\Sp$ while aligned dephasing contracts
$\Pc$, the composite diagnostic
\begin{equation}
\Aarrow:=\Sp+\lambda(1-\Pc^2),\qquad\lambda>0,
\label{eq:Aarrow}
\end{equation}
is nondecreasing.  This statement is deliberately conditional: dephasing alone preserves
the populations, and a population-mixing channel is an additional dynamical ingredient.

\begin{proposition}[Conditional monotonicity of $\Aarrow$]\label{prop:Aarrow}
Suppose that on an interval of evolution $\dot\Sp\geq0$ and $\dot\Pc^2\leq0$.
Then $\dot\Aarrow\geq0$ for every $\lambda>0$.
\end{proposition}
\begin{proof}
Direct differentiation gives $\dot\Aarrow=\dot\Sp-\lambda\dot\Pc^2\geq0$.
\end{proof}

Figure~\ref{fig:flow} and Table~\ref{tab:operational-dictionary} summarize the logical
content of the exact aligned-dephasing result.

\begin{figure*}[t]
\centering
\setlength{\fboxsep}{4pt}
\setlength{\fboxrule}{0.4pt}
\begin{minipage}{0.98\linewidth}
\centering
\setlength{\tabcolsep}{0pt}
\renewcommand{\arraystretch}{1.0}
\begin{tabular}{@{}c@{\hspace{0.25em}}c@{\hspace{0.25em}}c@{\hspace{0.25em}}c@{\hspace{0.25em}}c@{\hspace{0.25em}}c@{\hspace{0.25em}}c@{}}
\fbox{\parbox{0.190\linewidth}{\centering IRB-aligned\\
dephasing\\ (GKSL)}} &
$\Longrightarrow$ &
\fbox{\parbox{0.190\linewidth}{\centering Cohesion\\
contraction\\ $\Pc(t)\downarrow$}} &
$\Longrightarrow$ &
\fbox{\parbox{0.190\linewidth}{\centering Visibility\\
suppression\\ $V_{ij}(t)\downarrow$}} &
$\Longrightarrow$ &
\fbox{\parbox{0.190\linewidth}{\centering IRB-diagonal\\
description\\ within tolerance}}
\end{tabular}
\end{minipage}
\caption{Logical chain proved for phase-free IRB-aligned dephasing.  For general aligned
dephasing with diagonal phase evolution, the modulus-based functional $\Ctwo$ replaces
$\Pc$ when a phase-free interaction-picture representation is available.}
\label{fig:flow}
\end{figure*}

\begin{table*}[t]
\caption{Operational dictionary for the exact aligned-dephasing results.}
\label{tab:operational-dictionary}
\centering
\footnotesize
\renewcommand{\arraystretch}{1.25}
\begin{tabular}{lll}
\hline\hline
Observable & Formula or bound & Scope \\
\hline
Visibility & $V_{ij}=2|n_{ij}|/(a_i+a_j)$ & Phase-free IRB readout \\
Two-level visibility & $V=\Pc$ & Exact identity \\
Trace distance & $D_{\rm tr}(\rho_O,A)\leq\sqrt{d_{\rm act}/8}\,\Pc$ & Exact bound \\
Relative entropy & $C_{\rm rel}=S(\Delta_{\rm IRB}[\rho])-S(\rho)$ & Fixed-IRB benchmark \\
Classicalization time & $t_{\rm cl}^{\rm ub}=\Gamma_{\min}^{-1}\ln[\Pc(0)/\varepsilon]$ & Upper bound \\
\hline\hline
\end{tabular}
\end{table*}

\section{Coherence contraction under IRB-aligned dephasing}\label{sec:gksl}

\subsection{Exact dynamical class}

Fix the projectors $\{\Pi_i\}$ obtained from the IRB of a reference reduced state
$\rho(t_0)$.  Let
\begin{equation}
\AlgIRB:=\left\{\sum_i x_i\Pi_i:x_i\in\mathbb C\right\}
\end{equation}
be the commutative algebra diagonal in that IRB.  We call a GKSL generator
IRB-aligned dephasing if it admits a representation
\begin{align}
\mathcal L(\rho)=&-\mathrm{i}[H,\rho]\nonumber\\
&+\sum_\alpha\gamma_\alpha\left(L_\alpha\rho L_\alpha^\dagger
-\frac12\{L_\alpha^\dagger L_\alpha,\rho\}\right),
\label{eq:GKSL}
\end{align}
with $\gamma_\alpha\geq0$.
with~\cite{GKSL1976,Lindblad1976}
\begin{equation}
H=\sum_i h_i\Pi_i,\qquad
L_\alpha=\sum_i\ell_{\alpha i}\Pi_i .
\label{eq:Ldiag}
\end{equation}
This definition asserts alignment of the effective generator with the state-derived IRB;
it is not claimed to follow from the state alone.  It describes dephasing only: all
populations $\rho_{ii}$ are constant.  A channel that also mixes populations requires
additional off-diagonal jump operators and is outside the exact theorem below.

For $i\ne j$, Eqs.~\eqref{eq:GKSL}--\eqref{eq:Ldiag} give
\begin{align}
\dot\rho_{ij}&=-(\Gamma_{ij}+\mathrm{i}\Omega_{ij})\rho_{ij},\nonumber\\
\Gamma_{ij}&:=\frac12\sum_\alpha\gamma_\alpha
|\ell_{\alpha i}-\ell_{\alpha j}|^2\geq0.
\label{eq:Gammaij}
\end{align}
where $\Omega_{ij}$ contains the diagonal Hamiltonian frequency and any dissipative phase
shift.  Hence
\begin{equation}
|\rho_{ij}(t)|=|\rho_{ij}(t_0)|e^{-\Gamma_{ij}(t-t_0)} .
\label{eq:mod_decay}
\end{equation}

\subsection{Modulus-based contraction and intrinsic cohesion}

The phase-insensitive off-diagonal functional in the fixed IRB is
\begin{equation}
\Ctwo(t):=\sum_{i<j}|\rho_{ij}(t)|^2 .
\label{eq:Ctwo}
\end{equation}

\begin{proposition}[Contraction under IRB-aligned dephasing]\label{prop:Ctwo}
Under Eqs.~\eqref{eq:GKSL}--\eqref{eq:Ldiag},
\begin{equation}
\frac{\dd\Ctwo}{\dd t}
=-2\sum_{i<j}\Gamma_{ij}|\rho_{ij}(t)|^2\leq0 .
\label{eq:Ctwo_contract}
\end{equation}
If $\Gamma_{ij}>0$ for each initially nonzero off-diagonal element, then
$\Ctwo(t)\to0$.
\end{proposition}
\begin{proof}
Differentiate Eq.~\eqref{eq:Ctwo} and insert Eq.~\eqref{eq:Gammaij}; the phase
frequencies cancel from the derivative of each modulus squared.
\end{proof}

To track the intrinsic antisymmetric block itself, we restrict to the phase-free
subclass of IRB-aligned dephasing, meaning $\Omega_{ij}=0$ for the active pairs.  This condition
holds, for example, for Hermitian diagonal dephasing operators after transformation to
the interaction picture of the diagonal Hamiltonian.  Since the
reference IRB has $\rho_{ij}(t_0)=\mathrm{i}n_{ij}(t_0)$, one then has
\begin{equation}
N_{ij}(t)=N_{ij}(t_0)e^{-\Gamma_{ij}(t-t_0)},
\qquad \Ctwo(t)=\Ucoh(t).
\label{eq:nij_decay}
\end{equation}

\begin{proposition}[Cohesion Lyapunov law]\label{prop:PcContract}
For phase-free IRB-aligned dephasing,
\begin{equation}
\frac{\dd\Pc^2}{\dd t}
=-8\sum_{i<j}\Gamma_{ij}n_{ij}(t)^2\leq0 .
\label{eq:Pc_contract}
\end{equation}
Let $\Gamma_{\min}$ be the minimum positive $\Gamma_{ij}$ among initially nonzero
$N_{ij}$.  If no initially nonzero component is protected, then
\begin{equation}
\Pc(t)\leq\Pc(t_0)e^{-\Gamma_{\min}(t-t_0)} .
\label{eq:Ucoh_bound}
\end{equation}
\end{proposition}
\begin{proof}
Equation~\eqref{eq:Pc_contract} follows from Eq.~\eqref{eq:Pc} and
Eq.~\eqref{eq:nij_decay}.  Summing
$n_{ij}(t)^2\leq n_{ij}(t_0)^2e^{-2\Gamma_{\min}(t-t_0)}$ proves
Eq.~\eqref{eq:Ucoh_bound}.
\end{proof}

If some $\Gamma_{ij}=0$ on an initially nonzero component, complete contraction does not
occur.  The limiting object is then partially coherent or block-classical; this case is
treated in Sec.~\ref{sec:degeneracies}.

\subsection{Frame consistency and experimental alignment}
\label{sec:selfconsistency}

In the phase-free setting, $A$ is stationary and the fixed IRB remains an exact IRB for
all times.  In the Schr\"odinger picture with nonzero diagonal phase frequencies,
Eq.~\eqref{eq:mod_decay} still gives exact contraction of $\Ctwo$, but a purely imaginary
initial element can acquire a real part in the fixed frame.  Consequently, $\Pc$ as an
intrinsic antisymmetric descriptor is to be evaluated in the phase-free interaction
picture when it is used as the Lyapunov diagnostic.

\begin{proposition}[Fixed-frame bound with diagonal phases]\label{prop:frameStability}
Under IRB-aligned dephasing, if $\rho_{ij}(t_0)=\mathrm{i}n_{ij}(t_0)$, then the
real off-diagonal entries generated in the fixed initial IRB satisfy
\begin{equation}
\left|\operatorname{Re}\rho_{ij}(t)\right|
\leq |n_{ij}(t_0)|e^{-\Gamma_{ij}(t-t_0)} .
\label{eq:Sij_bound}
\end{equation}
Thus both the real and imaginary off-diagonal parts in the fixed frame vanish whenever
all active rates are positive.
\end{proposition}
\begin{proof}
Equation~\eqref{eq:mod_decay} bounds each real and imaginary component by the modulus of
$\rho_{ij}(t)$.
\end{proof}

Alignment is operationally testable: state tomography determines the reference IRB
projectors, and process tomography can determine whether the effective generator admits
an approximately diagonal representation in those projectors.  Failure of this alignment
does not contradict the IRB decomposition; it places the observed dynamics outside the
exact class treated here.

A sufficient microscopic route to alignment uses a fixed physical observable
$X=\sum_i x_i\Pi_i$ whose spectral projectors coincide with the candidate IRB sectors.
An interaction of the form
\begin{equation}
V_{SE}=\sum_\alpha f_\alpha(X)\otimes B_\alpha^E
\label{eq:VSE}
\end{equation}
produces dephasing diagonal in these sectors under the standard Markovian/secular
reduction when energy exchange can be neglected~\cite{BreuerPetruccione}.  Importantly, $X$ is a fixed physical
observable; the interaction is not taken to depend on the evolving reduced state.

Dynamics with small nonaligned perturbations can be studied by stability estimates, but
no quantitative perturbative bound is claimed in the present paper.

\section{Classicality criterion, classicalization time, and pointer structure}
\label{sec:threshold}

\subsection{Operational intrinsic-classicality criterion}

The criterion in this paper is relative to the fixed physical conjugation and to a
dephasing dynamics aligned with its IRB sectors.  It does not assert that a state is
classical with respect to every externally imposed measurement basis; it asserts that
its intrinsic cohesion is unresolved in the dynamically aligned sectors.  In the
phase-free setting, we call the state effectively IRB-classical at tolerance
$\varepsilon$ when
\begin{equation}
\Pc(t)\leq\varepsilon,\qquad 0<\varepsilon\ll1 .
\label{eq:threshold}
\end{equation}
By Eqs.~\eqref{eq:hsdist} and~\eqref{eq:trdist_bound}, this criterion entails
\begin{equation}
\big\|\rho_O-A\big\|_2\leq\frac{\varepsilon}{\sqrt{2}},
\qquad
D_{\rm tr}(\rho_O,A)\leq\sqrt{\frac{d_{\rm act}}{8}}\,\varepsilon .
\label{eq:threshold_dist}
\end{equation}
For general aligned dephasing with diagonal phase evolution, the phase-insensitive version
of the criterion is $2\sqrt{\Ctwo(t)}\leq\varepsilon$.

\subsection{Logarithmic classicalization-time bound}

Assume that every initially nonzero intrinsic cohesion component is damped with positive
rate and let $\Gamma_{\min}$ be the smallest such rate.  From
Eq.~\eqref{eq:Ucoh_bound}, the threshold~\eqref{eq:threshold} is reached no later than
\begin{equation}
 t_{\rm cl}^{\rm ub}(\varepsilon)=
 \max\!\left\{0,\frac{1}{\Gamma_{\min}}
 \ln\!\left(\frac{\Pc(t_0)}{\varepsilon}\right)\right\} .
\label{eq:tcl}
\end{equation}
This is an upper bound in the multi-rate problem; it becomes an equality if all initially
occupied components decay at $\Gamma_{\min}$, including the single-coherence two-level
case.

\subsection{Two-level example}

For a normalized two-level active sector, the IRB state is
\begin{equation}
\begin{gathered}
\rho_O=\begin{pmatrix}a&\mathrm{i}n\\-\mathrm{i}n&1-a\end{pmatrix},\\
0\leq a\leq1,\quad n^2\leq a(1-a) .
\end{gathered}
\label{eq:qubit}
\end{equation}
Here $\Pc=2|n|=V$, and phase-free dephasing gives
$\Pc(t)=\Pc(0)e^{-\Gamma t}$.  Therefore
\begin{equation}
 t_{\rm cl}(\varepsilon)=
 \max\!\left\{0,\Gamma^{-1}
 \ln\!\left(\frac{\Pc(0)}{\varepsilon}\right)\right\} .
\label{eq:tcl_qubit}
\end{equation}
Equivalently, one may replace $\Pc(0)$ by $V(0)$ in this two-level sector.
The equality $\Pc=V$ holds for any normalized two-level IRB state, not only for equal
populations.  The dynamics is illustrated in Fig.~\ref{fig:qubit_dynamics}.

\begin{figure}[t]
\centering
\includegraphics[width=0.80\linewidth]{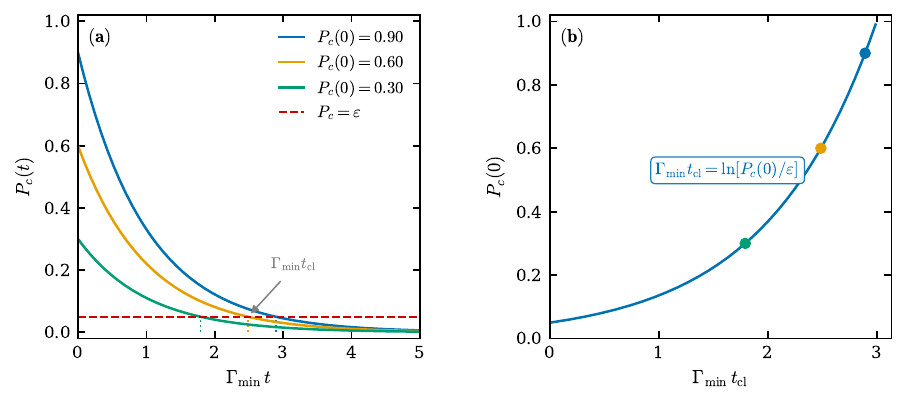}
\caption{Two-level IRB classicalization under a single phase-free dephasing rate.  The
cohesion index, identical here to fringe visibility, crosses the tolerance at
$\Gamma t_{\rm cl}=\ln[\Pc(0)/\varepsilon]$.}
\label{fig:qubit_dynamics}
\end{figure}

\subsection{Degeneracies and block-classicality}
\label{sec:degeneracies}

When $A$ has degenerate eigenvalues, the state alone identifies projectors
$P_\alpha$ onto degenerate IRB subspaces rather than preferred rank-one projectors within
each subspace.  Define the inter-block cohesion
\begin{equation}
\Pcperp^2:=4\sum_{\alpha<\beta}
\|P_\alpha N P_\beta\|_F^2 .
\label{eq:Pcperp}
\end{equation}
If the dephasing generator is scalar within each $P_\alpha$ but distinguishes different
blocks, then $\Pcperp$ contracts and the limiting state is block diagonal, whereas
cohesion inside a protected block may survive.  Thus block-classicality is not merely a
technical degeneracy case.  It describes a robust situation in which the environment
monitors coarse sectors while leaving intra-sector coherence unresolved.  States whose
cohesion is entirely internal to such unresolved blocks are fixed points of the
corresponding block-aligned dephasing: inter-block cohesion is suppressed, while
intra-block cohesion is retained.

If the environment distinguishes vectors inside a degenerate block, then the effective
generator selects a finer pointer structure that is not fixed by the state-derived IRB
alone.  Thus degeneracy makes explicit the complementary roles of state structure and
environmental monitoring.  The projectors $P_\alpha$ are state-derived candidate sectors;
their operational status as stable pointer sectors is a statement about the channel that
monitors or fails to monitor the coherences between and inside those sectors.

\subsection{Candidate pointer sectors and Born weights}
\label{sec:born}

The IRB projectors are candidate pointer sectors obtained from the state.  Under the
alignment hypothesis of Sec.~\ref{sec:gksl}, they are invariant under the dephasing
part of the generator; if the Hamiltonian is also diagonal in these projectors, they are
stable under the full aligned evolution.  They are selected as distinct operational
pointer sectors only when the inter-sector rates relevant to the initial state are
positive.  Once Eq.~\eqref{eq:threshold} is met, the
reduced state is operationally close to
\begin{equation}
A=\sum_i a_i\Pi_i,
\end{equation}
and the probabilities of readouts in these projectors are
$\Tr(\rho_O\Pi_i)=a_i$ by the standard trace rule.  The present construction does not
derive the Born rule; it identifies the intrinsic weights associated with the stable
sectors when the alignment condition holds.

\section{Relation to environment-induced einselection}\label{sec:comparison}

In environment-induced einselection, pointer states are determined by the dynamical
robustness imposed by the system--environment coupling~\cite{Zurek2003,Schlosshauer2007}.
The IRB does not replace that mechanism.  It provides, for a fixed physical conjugation,
a state-derived decomposition into populations and intrinsic cohesion, together with a
candidate family of projectors.  The exact result of this paper is conditional and
checkable: when the reduced dephasing generator is aligned with those candidate
projectors, intrinsic cohesion contracts at explicitly calculable rates and the
projectors are dynamically stable sectors.

This distinction is essential.  A reduced state can always be put into IRB form relative
to $K$, but the environment need not monitor its IRB sectors.  When state tomography and
process tomography show alignment, the channel-level classicalization bounds derived
above apply.  When they show persistent misalignment, one has identified a regime in
which the state-derived intrinsic decomposition and the interaction-selected pointer
structure are different.

Degeneracies sharpen the comparison.  A degenerate population block supplies only a
state-derived subspace.  Monitoring that is scalar inside the block leaves intra-block
cohesion protected, whereas monitoring that resolves the block chooses an additional
basis through the environment.  This is the same structural distinction encountered in
decoherence-free subspaces and noiseless codes~\cite{ZanardiRasetti1997,Lidar2014}.  In
this sense the block-classical regime identifies a channel-level mechanism by which a
coarse classical sector can coexist with protected coherence internal to that sector.

\section{Experimental signatures and falsifiability}\label{sec:signatures}

All exact quantities in the construction are accessible from reduced-state or process
tomography, and the two-level instance is accessible directly by interferometry.  An
experimental assessment proceeds in three stages.  First, tomography of $\rho(t_0)$ and
the declared physical conjugation determines the IRB projectors and $\Pc(t_0)$.  Second,
process tomography tests whether the inferred generator is approximately diagonal in
those projectors.  Third, if alignment holds, the measured decay envelopes are compared
with Eqs.~\eqref{eq:mod_decay}, \eqref{eq:Ucoh_bound}, and \eqref{eq:tcl}.

In a Ramsey or Mach--Zehnder implementation of a normalized two-level sector, the
direct prediction is
\begin{equation}
\begin{gathered}
V(t)=V(0)e^{-\Gamma t},\\
t_{\rm cl}(\varepsilon)=\Gamma^{-1}
\ln\!\left(\frac{V(0)}{\varepsilon}\right),\quad V(0)>\varepsilon .
\end{gathered}
\label{eq:tcl_ramsey}
\end{equation}
The visibility decay rate and the trace-distance decay rate obtained from tomography
must agree.  This benchmark can be implemented in photonic, trapped-ion, or
superconducting-qubit platforms~\cite{OBrien2009,DevoretSchoelkopf2013,Krantz2019}.

Multi-path implementations can test the block-classical extension.  If a degenerate
sector is not resolved by the generator, inter-block visibilities must decay while
protected intra-block visibilities may persist.  If a measured generator instead resolves
the nominally degenerate block, the observed finer pointer structure is attributed to the
dynamics rather than to the state-derived IRB.

The mechanism is falsified as an account of a given experiment if the generator is shown
to be aligned and Markovian in the stated sense while the off-diagonal moduli fail to
obey nonincreasing exponential envelopes with nonnegative rates.  Failure of alignment,
non-Markovian recoherence, or time-dependent support does not falsify the algebraic IRB
decomposition; it places the experiment beyond the exact dynamical class analyzed here.

\section{Discussion and outlook}\label{sec:discussion}

We have established an operational classicalization result for reduced quantum states in
a precisely delimited dynamical setting.  Relative to a fixed physical conjugation, the
IRB writes the state as $\rho_O=A+\mathrm{i}N$ and the globally normalized cohesion
index $\Pc^2=2\|N\|_F^2$ quantifies its intrinsic antisymmetric sector.  Under
phase-free dephasing aligned with the IRB projectors, $\Pc$ is a Lyapunov diagnostic and
contracts exponentially; under general aligned diagonal-phase evolution, the associated
modulus-based off-diagonal functional contracts exactly.  These results yield distance
bounds to the IRB-diagonal description and a logarithmic upper bound on the time required
to reach a prescribed tolerance.

The scope of the statement is intentionally explicit.  The IRB supplies candidate pointer
sectors from the reduced state, while their dynamical stability requires alignment with
the effective generator.  This is not an alternative derivation of einselection from the
state alone; it is a state-and-channel criterion that can be verified by tomography.  In
the degenerate case it predicts block-classicality when the environment monitors only
coarse sectors and allows protected intra-block cohesion when the generator does not
resolve those sectors.

The two-level case makes the criterion directly observable because $\Pc$ equals fringe
visibility.  In higher-dimensional systems, pairwise visibilities and tomography provide
complementary access to the cohesion components and to the distance bounds.  The relative
entropy of coherence remains an appropriate general resource-theoretic comparator for a
fixed IRB projective structure; the present paper claims Lyapunov behavior of $\Pc$ only
for the aligned phase-free class for which it has been proved.

Three extensions are natural.  One is a controlled treatment of slowly varying intrinsic
projectors when populations evolve appreciably during decoherence.  The second is a
systematic comparison between the IRB candidate sectors and pointer structures obtained
from explicit microscopic environments.  The third is to ask whether block structures of
the kind analyzed here can be selected by variational, calibration, or information-geometric
principles, rather than imposed by an externally reconstructed generator.  Such a principle
would determine the physical conjugation or block structure before the present Lyapunov
statements are applied.  None of these extensions is required for the exact contraction
and classicalization-time bounds derived here.

\appendix

\section{Intrinsic reference basis: construction and uniqueness}
\label{app:IRB}

Given a fixed physical conjugation $K$, any Hermitian density operator on the active
subspace admits the decomposition
\begin{align}
\rho&=S+\mathrm{i}T,\nonumber\\
S&:=\frac{\rho+K\rho K}{2},\quad S^{\mathsf T}=S,\nonumber\\
T&:=\frac{\rho-K\rho K}{2\mathrm{i}},\quad T^{\mathsf T}=-T.
\label{eq:SAdecomp_app}
\end{align}
Choose $Q\in SO(d_{\rm act})$ so that
$Q^{\mathsf T}SQ=A=\operatorname{diag}(a_1,\ldots,a_{d_{\rm act}})$ with ordered
entries, and define $N=Q^{\mathsf T}TQ$.  This yields
$\rho_O=A+\mathrm{i}N$.

If the spectrum of $S$ is nondegenerate, ordering fixes the permutation freedom, but
projectors are insensitive to eigenvector signs and $Q$ retains orientation-preserving
sign choices.  These choices may change signs of individual $n_{ij}$ while leaving
$\Pc$, $\Ucoh$, and the projectors invariant.  If $S$ is degenerate, $Q$ retains
orthogonal freedom within each degenerate eigenspace; the canonical state-derived objects
are then the block projectors and block-invariant quantities such as
Eq.~\eqref{eq:Pcperp}.

The conjugation $K$ is part of the physical specification of the analysis.  It may be
set by a laboratory convention, by a real symmetry representation, or by a secular
basis.  It is fixed before assessing cohesion decay.  Equivalently, if a further physical
or information-theoretic rule is used to choose $K$ from the state and its context, that
choice is made first and becomes part of the channel specification.  The further statement
that the resulting IRB sectors are pointer sectors is not automatic: it is the alignment
condition on the reduced generator tested in Sec.~\ref{sec:selfconsistency}.

\section{Two-path interferometric visibility benchmark}
\label{app:visibility}

In the two-dimensional subspace spanned by $\{\lvert i\rangle,\lvert j\rangle\}$, an
ideal balanced interferometric readout produces
\begin{equation}
p(\phi)=\frac{\rho_{ii}+\rho_{jj}}{2}
+\operatorname{Re}\!\left(\rho_{ij}e^{\mathrm{i}\phi}\right).
\label{eq:p_phi_app}
\end{equation}
The standard visibility is
\begin{equation}
V_{ij}=\frac{p_{\max}-p_{\min}}{p_{\max}+p_{\min}}
=\frac{2|\rho_{ij}|}{\rho_{ii}+\rho_{jj}} .
\label{eq:Vij_app}
\end{equation}
In a phase-free IRB description, $\rho_{ij}=\mathrm{i}n_{ij}$, yielding
Eq.~\eqref{eq:visibility_general}.  If the active Hilbert space is two dimensional and
normalized, $a_1+a_2=1$ and therefore $V_{12}=2|n_{12}|=\Pc$ by
Eq.~\eqref{eq:Pc}; population balance is not required.  With diagonal phase evolution,
the visibility follows the modulus contraction of Eq.~\eqref{eq:mod_decay}.

\section*{Author contribution statement}
Language-editing and editorial-structuring assistance from AI-based tools was used during
manuscript preparation. All scientific content, derivations, and final decisions are the
sole responsibility of the author.

\end{document}